\documentclass{article}

\usepackage[T1]{fontenc}
\usepackage{amsmath}
\usepackage{mathrsfs}
\usepackage{amstext}
\usepackage{amsthm}
\usepackage{amssymb}
\usepackage[unicode=true,pdfusetitle,
 bookmarks=true,bookmarksnumbered=false,bookmarksopen=false,
 breaklinks=false,pdfborder={0 0 1},backref=false,colorlinks=false]
 {hyperref}
\usepackage{color}
\usepackage{float}
\usepackage{graphicx}
\usepackage[]{algorithm2e}

\numberwithin{equation}{section}
\numberwithin{figure}{section}
\theoremstyle{plain}
\newtheorem{thm}{\protect\theoremname}
  \theoremstyle{definition}
  \newtheorem{defn}[thm]{\protect\definitionname}
  \theoremstyle{plain}
  \newtheorem{lem}[thm]{\protect\lemmaname}
  \theoremstyle{remark}
  \newtheorem{rem}[thm]{\protect\remarkname}
  \theoremstyle{definition}
  \newtheorem{example}{\protect\examplename}
  \theoremstyle{plain}
  \newtheorem{prop}[thm]{\protect\propositionname}

\makeatother

  \providecommand{\definitionname}{Definition}
  \providecommand{\examplename}{Example}
  \providecommand{\lemmaname}{Lemma}
  \providecommand{\propositionname}{Proposition}
  \providecommand{\remarkname}{Remark}
\providecommand{\theoremname}{Theorem}

\begin{document}

\title{Representing Model Ensembles as Boolean Functions}
\author{Robert Schwieger, Heike Siebert\\
Department of Mathematics, Freie Universität Berlin, Germany}

\maketitle

\global\long\def\glIG#1{IG^{\text{global}}(#1)}
\global\long\def\IG#1#2{IG_{#1}(#2)}
\global\long\def\sG#1{G_{\text{sync}}(#1)}
\global\long\def\sV#1{V_{\text{sync}}(#1)}
\global\long\def\sE#1{E_{\text{sync}}(#1)}
\global\long\def\aG#1{G_{\text{async}}(#1)}
\global\long\def\aV#1{V_{\text{async}}(#1)}
\global\long\def\aE#1{E_{\text{async}}(#1)}
\global\long\def\mG#1{G_{\text{QDE}}(#1)}
\global\long\def\mV#1{V_{\text{QDE}}(#1)}
\global\long\def\mE#1{E_{\text{QDE}}(#1)}
\global\long\def\pfi#1{\phi_{#1}}
\global\long\def\apfi#1{\tilde{\phi}_{#1}}
\global\long\def\redA#1{\tilde{G}_{#1}}
\global\long\def\redV#1{\tilde{V}_{#1}}
\global\long\def\redE#1{\tilde{E}_{#1}}
\global\long\def\dqG{G_{\text{QDE}}^{\text{Boolean}}}
\global\long\def\dqV{V_{\text{QDE}}^{\text{Boolean}}}
\global\long\def\dqE{E_{\text{QDE}}^{\text{Boolean}}}
\global\long\def\Z#1#2{\text{diff}(#1,#2)}
\global\long\def\supp#1#2{\text{comm}(#1,#2)}
\global\long\def\M{\mathscr{M}}
\global\long\def\MB{\M_{\mathbb{B}}}
\global\long\def\cMB{\overline{\M}_{\mathbb{B}}}

\global\long\def\z{\delta}
\global\long\def\invz{{\color{red}\Delta}}
\global\long\def\l{\varphi}
\global\long\def\p{\phi}
\global\long\def\P{\Phi}
\global\long\def\g{g}
\global\long\def\T{T}
\global\long\def\R{R}
\global\long\def\mux{\mu^{\Sigma}}
\global\long\def\pZ#1{\text{diff}(#1)}
\global\long\def\psupp#1{\text{comm}(#1)}


\section*{Abstract}
\vspace{-2mm}
Families of ODE models $\dot{x}=f(x)$ characterized by a common sign
structure $\Sigma$ of their Jacobi matrix $J(f)$ are investigated
within the formalism of qualitative differential equations \cite{Eisenack2006}.
In the context of regulatory networks the sign structure of the Jacobi matrix carries the information
about which components of the network inhibit or activate each other.
Information about constraints on the behavior of models in this family
is stored in a so called qualitative state transition graph $\mG{\Sigma}$
\cite{Eisenack2006}. We showed previously that a similar approach
can be used to analyze a model pool of Boolean functions characterized
by a common interaction graph \cite{SchwiegerMonotonicModelPools}.
Here we show that the opposite approach is fruitful as well. We show
that the qualitative state transition graph $\mG{\Sigma}$ can be
reduced to a ``skeleton'' represented by a Boolean function $f^{\Sigma}$
conserving the reachability properties. This reduction
has the advantage that approaches such as model checking
and network inference methods can be applied to $\aG{f^{\Sigma}}$
within the framework of Boolean networks. Furthermore, our work constitutes an alternative to approaches like
\cite{snoussi_1989} and \cite{wittmann} to link
Boolean networks and differential equations.

\vspace{-2mm}
\section{Introduction}

Mathematical modeling in systems biology is often hampered by lack
of information on mechanistic detail and parameters. Several approaches
deal with this problem. Here we focus on the theory of qualitative
differential equations (QDE). A qualitative differential equation
model is an abstraction of a system of ordinary differential equation,
consisting of a set of real-valued variables and functional, algebraic
and differential constraints among them \cite{Kuipers2001}. In the
simplest case the objects of interest are systems of ordinary differential
equations (ODEs) consistent with a given signed interaction graph
$\Sigma\in\{-1,0,1\}^{n\times n}$, $n\in\mathbb{N}$ capturing dependencies
between system components and the type of influence exerted, activating
or inhibiting \cite{Eisenack2006}. These ODEs are collected in a
so called model ensemble. For this ensemble, a qualitative state transition
graph (QSTG) $\mG{\Sigma}$ can be constructed whose nodes represent
derivative signs of the system components and edges indicate possible
changes in the derivative over time. It can then be used to describe
the behavior of the ensemble. This will be explained in Section \ref{subsec:QDEgraph}. This constitutes a scenario of particular
interest in application, where interaction information is usually
more readily available than details on the processing logic of multiple
influences on a target component.

In this paper we show that the graph $\mG{\Sigma}$ can be reduced
to an asynchronous state transition graph of a Boolean function $f^{\Sigma}$.
This allows us to study the graph $\mG{\Sigma}$ with existing tools
for Boolean regulatory networks and to use theoretical results about
Boolean regulatory networks to analyze it.

Our paper is structured in the following way: In the first section
we state definitions and notions about Boolean regulatory networks.
Afterwards, we review existing results for monotonic ensembles in
the continuous setting and define a Boolean version of the qualitative
state transition graph of such an ensemble denoted by $\mG{\Sigma}$.
In Section \ref{sec:Skeleton-of-a} we introduce the skeleton of the
graph $\mG{\Sigma}$ and prove that no information about reachability
is lost during this reduction. Subsequently, we exploit the results
by using model checking and outline how the result could be used for
network inference.

\vspace{3mm}
\subsection{Boolean networks}

We denote with $f$ a Boolean function $\{0,1\}^{n}\rightarrow\{0,1\}^{n}$
and with $n\in\mathbb{N}$ the dimension of its state space. With
$[n]$ we denote the set $\{1,\dots,n\}$. For $v\in \{0,1\}^{n}, i\in [n]$ the value $v_i$ refers to the $i$-th component of the network.

Furthermore, we define for $v,w\in\{0,1\}^{n}$ the following sets:
\begin{defn}
\label{def:DiffAndComm}For
$v,w\in\{0,1\}^{n}$ the set $\Z vw$ contains all indices where $v$
and $w$ are different and the set $\supp vw$ the set of indices
where $v$ equals $w$:
\begin{align*}
\Z vw & :=\big\{ i\in[n]|v_{i}\not=w_{i}\big\},\\
\supp vw & :=\big\{ i\in[n]|v_{i}=w_{i}\big\}.
\end{align*}
\end{defn}
For a vector $v\in\{0,1\}^{n}$ the vector $v^{A}$ denotes its negation on the set of components $A\subset[n]$. More precisely:
\begin{defn}
\label{def:v_hoch_A}Let
$A\subset[n]$ and $v\in\{0,1\}^{n}$. Define $v^{A}\in\{0,1\}^{n}$
component wise like this:
\[
v_{i}^{A}=\begin{cases}
v_{i} & \text{if }i\not\in A\\
\neg v_{i} & \text{if }i\in A
\end{cases}.
\]
\end{defn}

We attribute to a Boolean function $f:\{0,1\}^{n}\rightarrow\{0,1\}^{n}$
a relation on $\{0,1\}^{n}$ in the following way:
\begin{defn}
\label{def:ASTG}We attribute to $f:\{0,1\}^{n}\rightarrow\{0,1\}^{n}$
an asynchronous state transition graph (ASTG)
$\aG f=\big(\aV f,\aE f\big)$ with
\[
\aV f:=\{0,1\}^{n}
\]
and 
\begin{align*}
\aE f & =\big\{(s,t)\in\aV f\times\aV f\big|\big(\{i\}=\Z st\\
 & \text{ and }f_{i}(s)=t_{i}\big)\text{ or }s=t=f(s)\big\}.
\end{align*}
\end{defn}
The ASTG captures a relation on the the set $\{0,1\}^{n}$. We can
describe this relation also with a function $\mu:\{0,1\}^{n}\rightarrow\{0,1\}^{n}$
such that for $i\in[n]$
\begin{align*}
(v,v^{\{i\}})\in\aE f & \Leftrightarrow\mu_{i}(v)\\
\end{align*}
holds.  To keep our notation simple and since we can identify a function $\{0,1\}^n \rightarrow \{0,1\}$ with a logical formula, we write $\mu_{i}(v)$ instead of $\mu_{i}(v) = 1$.
For $a,b\in\{0,1\}$ we define $a\oplus b:=a\text{ XOR }b:=\big(\neg a\wedge b\big)\vee\big(a\wedge\neg b\big)$.
Lemma~\ref{lem:ConditionFunctionEquivalence} describes how
we can obtain $\mu$ from $f$ and vice versa.

\begin{lem}
\label{lem:ConditionFunctionEquivalence}Assume $f:\{0,1\}^{n}\rightarrow\{0,1\}^{n}$.
Then for $v,v^{\{i\}}\in\{0,1\}^{n}$ and $i\in[n]$ it holds

\[
(v,v^{\{i\}})\in\aE f\Leftrightarrow\mu_{i}(v)
\]
and 
\[
(v,v)\in\aE f\Leftrightarrow\forall i\in[n]:\neg\mu_{i}(v)
\]
with $\mu_{i}(v):=v_{i}\oplus f_{i}(v)$.
\end{lem}
\begin{proof}
Consider
\begin{align*}
(v,v^{\{i\}})\in\aE f\Leftrightarrow & i\in\Z v{f(v)}\\
\Leftrightarrow & v_{i}\oplus f_{i}(v)
\end{align*}
and 
\begin{align*}
(v,v)\in\aE f & \Leftrightarrow\forall i\in[n]:v_{i}=f_{i}(v)\\
 & \Leftrightarrow\forall i\in[n]:\neg(v_{i}\oplus f_{i}(v))
\end{align*}
\end{proof}
Expressing the relation represented by the graph $\aG f$ via a Boolean
function $\mu$ will turn out useful in the sequel. Similar to Lemma
\ref{lem:ConditionFunctionEquivalence} we can represent a graph constructed
from a function $\mu:\{0,1\}^{n}\rightarrow\{0,1\}^{n}$ as an ASTG
of a Boolean function:
\begin{rem}
Lemma~\ref{lem:ConditionFunctionEquivalence} implies that for a graph
$G=(V,E)$ defined by $V=\{0,1\}^{n}$ and
\begin{align*}
(v,v^{\{i\}}) & \in E:\Leftrightarrow\mu_{i}(v),\\
(v,v) & \in E:\Leftrightarrow\forall i\in[n]:\neg\mu_{i}(v)
\end{align*}
for $i\in[n]$ and $\mu:\{0,1\}^{n}\rightarrow\{0,1\}^{n}$ it holds
$\aG{\mu\oplus id}=G$.
\end{rem}
\begin{proof}
Choose $f=\mu\oplus id$ in Lemma~\ref{lem:ConditionFunctionEquivalence}.
\end{proof}

\vspace{3mm}
\subsection{\label{subsec:QDEgraph}The graph $\protect\mG{\Sigma}$}

Families of ODE models $\dot{x}=\overline{f}(x)$ characterized by
a common sign structure $\Sigma=(\sigma_{i,j})_{i,j\in[n]}\in\{-1,0,1\}^{n\times n}$
in its Jacobi matrix $J(\overline{f})$ can be investigated using
qualitative differential equations \cite{Eisenack2006}. Such families
of ODE models are called monotonic ensembles. Instead of the solutions
$x(\cdot)$ of the ODE-systems, so-called ''abstractions'' are considered.
Here, these abstractions are sequences of sign vectors of the derivatives
of the solutions. A state transition graph $\mG{\Sigma}$ on the sign
vectors can be constructed based on the sign matrix $\Sigma$, which
captures restrictions on the behavior of the solutions. We give here
a short review of the construction of $\mG{\Sigma}$. Since we are
here only interested in the properties of the object $\mG{\Sigma}$,
most definitions are skipped, but illustrated with an example. For
details and exact definitions we refer instead to \cite[Sec. 2.2.]{SchwiegerMonotonicModelPools}
and \cite[Chapter 2.1-2.2]{Eisenack2006}. For understanding the following
sections nothing more than the definition of $\mG{\Sigma}$ is necessary
which will be given in the end of this section.

We define an ensemble of ODE systems $\M(\Sigma)$ whose corresponding
Jacobi matrices share a sign structure. The usual sign operator is
denoted $[\cdot]:=sign(\cdot)$. We abstract the solutions of the
ODE-systems in the model ensemble to sequences of sign vectors that
describe the slope of their derivatives. 
\begin{example}
Let us assume our monotonic ensemble $\M(\Sigma)$ is characterized
by the sign structure $\Sigma=\begin{pmatrix}-1 & 0 & 0 & -1\\
1 & -1 & 0 & 0\\
0 & 1 & -1 & -1\\
0 & 0 & -1 & -1
\end{pmatrix}$. As demonstrated in \cite[Example 2]{SchwiegerMonotonicModelPools}
an example of such a function would be $\overline{f}(x)-x$ with $\overline{f}:[0,1]^{4}\rightarrow[0,1]^{4},(x_{1},x_{2},x_{3},x_{4})\mapsto(1-\frac{x_{4}}{x_{4}+0.5},\frac{x_{1}}{x_{1}+0.5},1-\frac{x_{4}}{x_{4}+0.5},1-\frac{x_{3}}{x_{3}+0.5})$.
I.e. $\big(\overline{f}(x)-x\big)\in\M(\Sigma)$ since it can be easily
shown that the abstraction of the Jacobi matrix of $\big(\overline{f}(x)-x\big)$
on $[0,1]^{n}$ is exactly $\Sigma$, i.e. $[J(\overline{f}(x)-x)]=\Sigma$.
In Fig.~\ref{fig:ExampleModelEnsemble} the solution of the corresponding
ODE system $\dot{x}(t)=\big(\overline{f}(x(t))-x(t)\big),t\in\mathbb{R}_{\geq0}$
with initial value $x_{0}=(0.6,0.6,0.6,0.6)$ is given. When we consider
only signs of slopes of this solution we obtain the trajectory of
sign vectors $(-1,-1,-1,-1)\rightarrow(1,-1,-1,-1)$$\rightarrow(1,1,-1,-1)$.
By taking the union of all such abstractions of the solutions of all
the ODE systems in the model ensemble $\M(\Sigma)$ we obtain the
graph $\mG{\Sigma}=\big(\{-1,1\}^{n},\mE{\Sigma}\big)$. However,
according to \cite[p. 25]{Eisenack2006}, \cite[Proposition 1]{SchwiegerMonotonicModelPools}
there is a way to construct the graph $\mG{\Sigma}$ without solving
any ODE.

\begin{figure}
\includegraphics[width=0.9\textwidth]{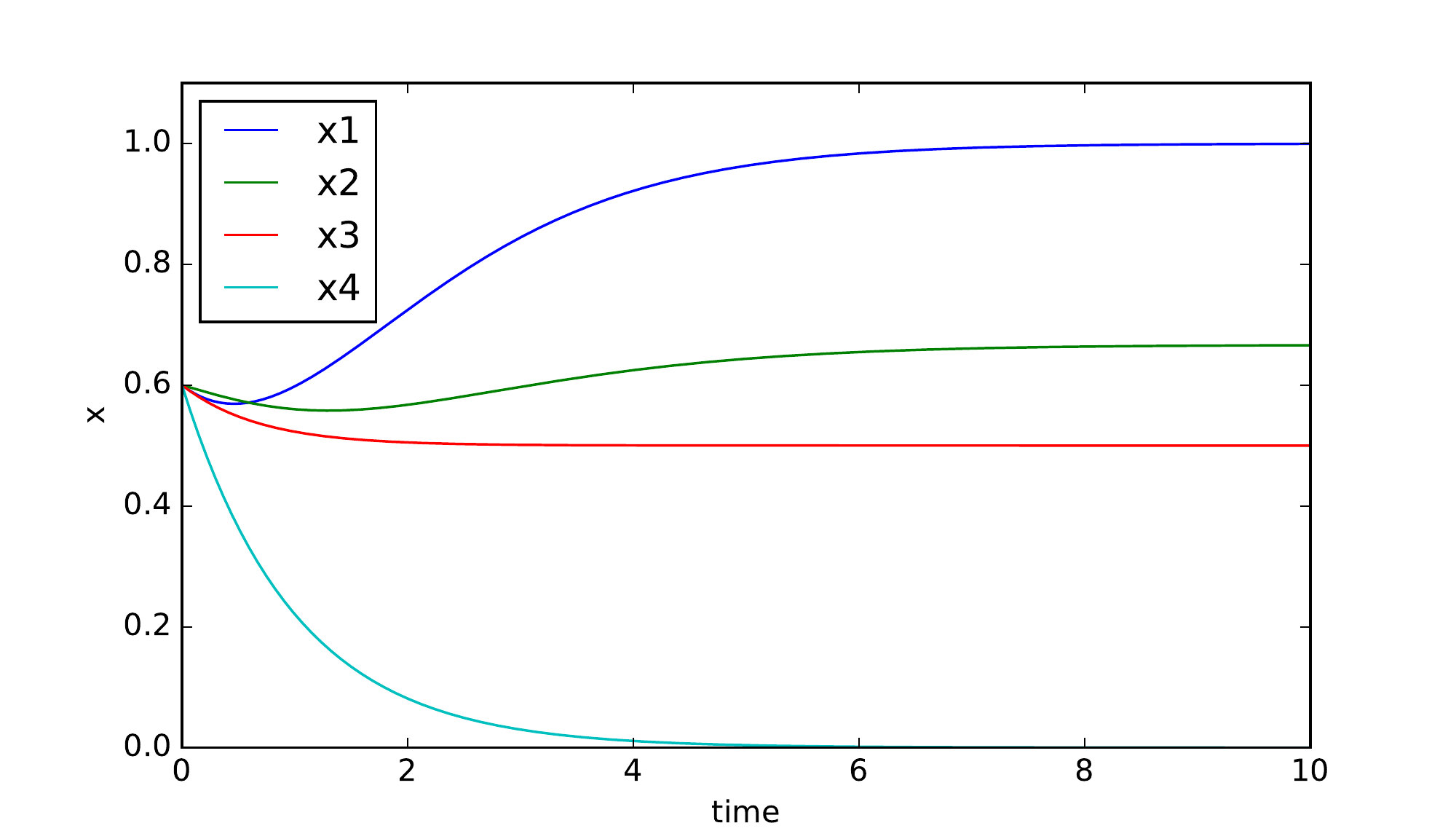}\caption{\label{fig:ExampleModelEnsemble} Trajectories of a solution of an
ODE in the model ensemble. Its abstraction is given by $(-1,-1,-1,-1)\rightarrow(1,-1,-1,-1)$$\rightarrow(1,1,-1,-1)$.
The ODE system was parametrized with $x_{0}=(0.6,0.6,0.6,0.6)$.}

\end{figure}
\end{example}
This result motivates the following graph:
\begin{defn}
For a sign matrix $\Sigma=(\sigma_{i,j})_{i,j\in[n]}\in\{-1,0,1\}^{n\times n}$we
define the graph $\overline{\mG{\Sigma}}=\big(\overline{\mV{\Sigma}},\overline{\mE{\Sigma}}\big)$
by
\[
\overline{\mV{\Sigma}}=\{-1,1\}^{n}
\]
 and
\begin{align}
(v,w)=e\in\overline{\mE{\Sigma}}\nonumber \\
:\Leftrightarrow\forall i\in\Z vw\exists j\in\supp vw:\big(v_{i}\cdot v_{j}\cdot(-1)\equiv\sigma_{i,j}\big)\label{eq:QDEcondition}
\end{align}
\end{defn}
The graph $\overline{\mG{\Sigma}}$ describes the dynamic restrictions on the solutions of the ODEs in the model ensemble imposed
by the sign matrix $\Sigma$.
Since we are here only interested in properties of the graph $\overline{\mG{\Sigma}}$,
we refer for details about the relation between the qualitative state
transition graph and its model ensemble to \cite{Eisenack2006}. Furthermore,
we note that there is no one-to-one correspondence between the qualitative
state transition graph and the corresponding sign matrix $\Sigma$.
It is possible to change elements on the diagonal of $\Sigma$ without
changing the graph $\overline{\mG{\Sigma}}$. This is due to the fact
that the sets $\Z vw$ and $\supp vw$ are disjoint and thus the diagonal
elements do not play a role in (\ref{eq:QDEcondition}). Consequently,
the edge set does not change when changing the diagonal of $\Sigma$.

Since we are interested in the relation of the QSTG and Boolean networks,
we redefine now the graph $\overline{\mG{\Sigma}}$ on the node set
$\{0,1\}^{n}$. 
\begin{defn}
\label{def:QDEGraph}For
a sign matrix $\Sigma=(\sigma_{i,j})_{i,j\in[n]}\in\{-1,0,1\}^{n\times n}$we
define the graph $\mG{\Sigma}=\big(\mV{\Sigma},\mE{\Sigma}\big)$
by 
\[
\mV{\Sigma}=\{0,1\}^{n}
\]
and
\begin{align}
(v,w)=e\in\mE{\Sigma}\big)\nonumber \\
:\Leftrightarrow\forall i\in\Z vw\exists j\in\supp vw:\big[\sigma_{i,j}\not=0\wedge\big(v_{i}\oplus^{\sigma_{i,j}}v_{j}\big),\label{eq:QDEcondition-1}
\end{align}
where $a\oplus^{-1}b:=\neg(a\oplus b)$ and $a\oplus^{1}b:=a\oplus b$
for $a,b\in\{0,1\}$.
\end{defn}
From the definition it is clear that the graph $\mG{\Sigma}$ and
$\overline{\mG{\Sigma}}$ are the same after relabeling $0$ to $-1$.
It is easy to see that conditions~(\ref{eq:QDEcondition-1}) and (\ref{eq:QDEcondition})
correspond to each other.
\begin{example}
In Fig.~\ref{fig:QDE_graph_running_example} the graph $\mG{\Sigma}$
from the previous example is depicted. We can see from this for example
that no solution in the model ensemble has a direct transition from
$(-1,1,1,-1)$ to $(-1,1,1,1)$, since there is no edge $(0,1,1,0)\rightarrow(0,1,1,1)$
in the graph. 
\begin{figure}[t]
\includegraphics[width=0.9\textwidth]{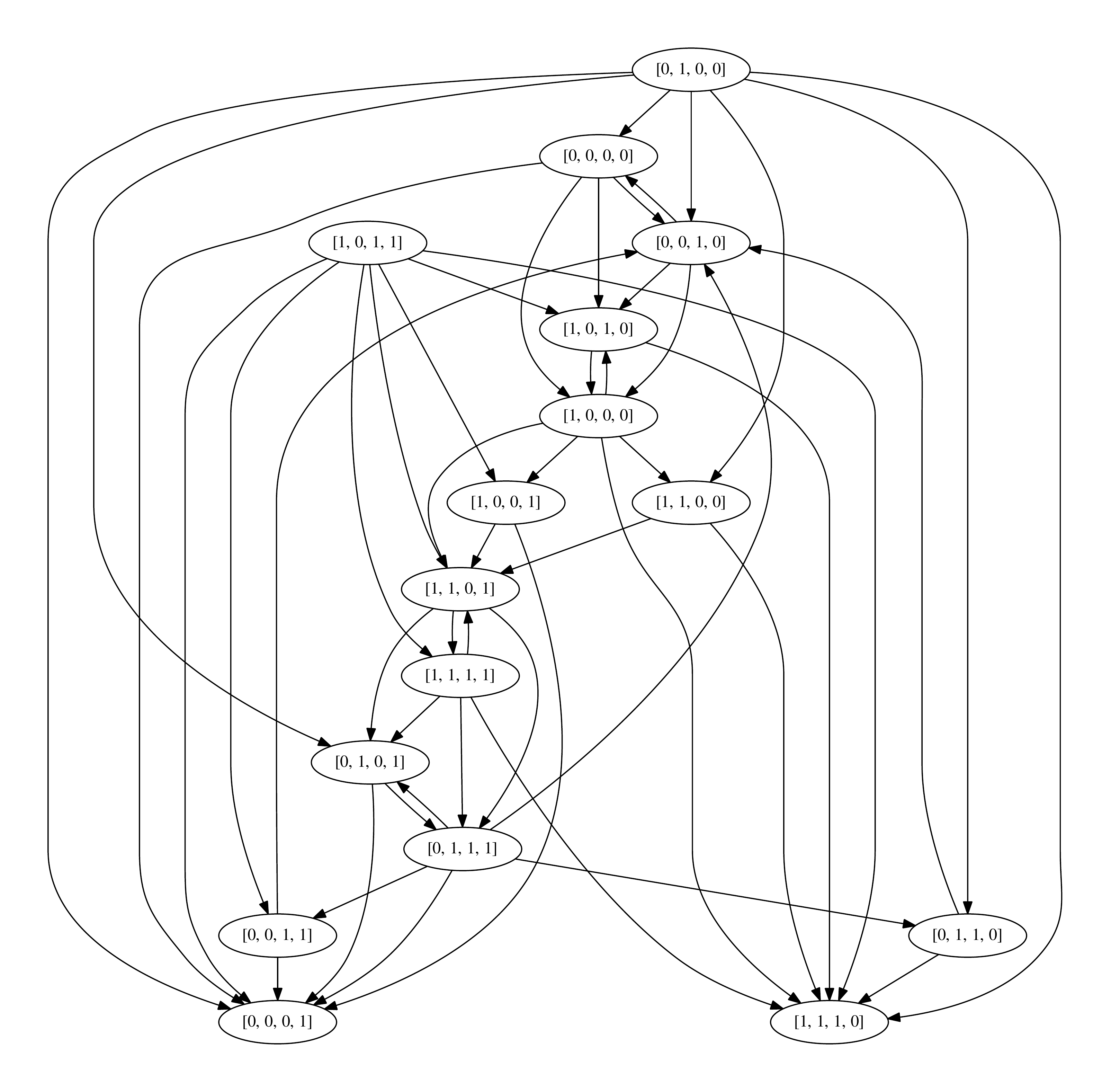}\caption{\label{fig:QDE_graph_running_example}The graph $\protect\mG{\Sigma}$
of the running example.}

\end{figure}
\end{example}

\vspace{3mm}
\section{\label{sec:Skeleton-of-a}Skeleton of a Model Ensemble}

In the previous section we defined the graphs $\mG{\Sigma}$ in such
a way that for an edge $e=(v,w)$, $v,w\in\{0,1\}^{n}$ with $\Z vw=A$
the following conditions hold:
\[
\forall i\in A\exists j\in A^{c}:\big(\sigma_{i,j}\not=0\wedge\big(v_{i}\oplus^{\sigma_{i,j}}v_{j}\big)\big)\Leftrightarrow e\in\mE{\Sigma}.
\]
This implies that there are potentially edges $(v,w)\in\mE{\Sigma}$
where the set $\Z vw$ has cardinality bigger than one. However, we
want to represent the graph $\mG{\Sigma}$ as an ASTG of a Boolean
function. Therefore, we show now in the sequel that we can delete
the edges $(v,w)$ with $| \Z vw |>1$ without loosing information about
reachability. $| \cdot |$ denotes the cardinality of a set. 

For the sets $A\subseteq[n]$ with cardinality one ($|A|=1$) these
conditions are given by the following logical formula:
\begin{align}
\mux_{i}(v)= & \exists j\in[n]\backslash\{i\}:\big(\sigma_{i,j}\not=0\wedge\big(v_{i}\oplus^{\sigma_{i,j}}v_{j}\big)\big)\nonumber \\
= & \vee_{j\in[n]\backslash\{i\}\text{ s.t. }\sigma_{i,j}\not=0}\big(v_{i}\oplus^{\sigma_{i,j}}v_{j}\big).\label{eq:conditionFunction}
\end{align}
We prove now that we can restrict the graph $\mG{\Sigma}$ to its
edges of the form $(v,v^{\{i\}})$, $v\in\{0,1\}^{n}$, $i\in[n]$
without loosing information about reachability.
\begin{prop}
\label{prop:ClosednessQDE}Assume there is an edge $(v,v^{A})\in\mE{\Sigma}$
with $A\subseteq[n]$, $A\not=\emptyset$. Then for each $B\subseteq A$,
$A\not=B\not=\emptyset$:
\begin{align*}
(v,v^{B}) & \in\mE{\Sigma},\\
(v^{B},v^{A}) & \in\mE{\Sigma}.
\end{align*}
\end{prop}
\begin{proof}
We show first $(v,v^{B})\in\mE{\Sigma}$: Since $B\subseteq A\Leftrightarrow A^{c}\subseteq B^{c}$
\begin{align*}
(v,v^{A})\in\mE{\Sigma}\\
\Rightarrow\forall i\in A\exists j\in A^{c}:\big(\sigma_{i,j}\not=0\wedge\big(v_{i}\oplus^{\sigma_{i,j}}v_{j}\big)\big)\\
\Rightarrow\forall i\in B\exists j\in B^{c}:\big(\sigma_{i,j}\not=0\wedge\big(v_{i}\oplus^{\sigma_{i,j}}v_{j}\big)\big)\\
\Rightarrow(v,v^{B})\in\mE{\Sigma}
\end{align*}
For the second part we know due to $(v,v^{A})\in\mE{\Sigma}$:
\begin{align*}
\forall i\in A\exists j\in A^{c}:\big(\sigma_{i,j}\not=0\wedge\big(v_{i}\oplus^{\sigma_{i,j}}v_{j}\big)\big)\\
\end{align*}
Let us call $C:=\Z{v^{B}}{v^{A}}=A\backslash B$. We need to show
\begin{align*}
\forall i\in C\exists j\in C^{c}:\big(\sigma_{i,j}\not=0\wedge\big(v_{i}^{B}\oplus^{\sigma_{i,j}}v_{j}^{B}\big)\big)\\
\end{align*}

Before we start by noting two observations.\\
\textbf{1. Observation}: $\forall j\in A^{c}:v_{j}=v_{j}^{A}=v_{j}^{B}$\\
\\
\textbf{2. Observation}: $\forall i\in C:v_{i}=v_{i}^{B}$\\
\\
\\
Now we are ready to prove the statement of the Proposition.

Since $B\bigcup C=A$ we have $C\subseteq A$. It follows
\begin{align*}
\forall i\in A\exists j\in A^{c}:\big(\sigma_{i,j}\not=0\wedge\big(v_{i}\oplus^{\sigma_{i,j}}v_{j}\big)\big)\\
\Rightarrow\forall i\in C\exists j\in A^{c}:\big(\sigma_{i,j}\not=0\wedge\big(v_{i}\oplus^{\sigma_{i,j}}v_{j}\big)\big)\\
\end{align*}
And due to Observation 1 and 2 we obtain:
\begin{align*}
\Rightarrow\forall i\in C\exists j\in A^{c}:\big(\sigma_{i,j}\not=0\wedge\big(v_{i}^{B}\oplus^{\sigma_{i,j}}v_{j}^{B}\big)\big)\\
\end{align*}
Since $C\subseteq A\Leftrightarrow A^{c}\subseteq C^{c}$ finally
\begin{align*}
\Rightarrow\forall i\in C\exists j\in C^{c}:\big(\sigma_{i,j}\not=0\wedge\big(v_{i}^{B}\oplus^{\sigma_{i,j}}v_{j}^{B}\big)\big) & \Leftrightarrow(v^{B},v^{A})\in\mE{\Sigma}\\
\end{align*}
\end{proof}
Proposition~\ref{prop:ClosednessQDE} shows that we can restrict ourself
to the edges induced by the sets $\{1\},\dots,\{n\}$ without loosing
information about reachability of nodes in $\mG{\Sigma}$. This graph
can have significantly less edges. Since each edge is of the form $(v,v^{\{i\}})$,
it should be possible to represent this graph as the ASTG of a suitable
Boolean function, which we will call $f^{\Sigma}$. Indeed, due to
Lemma~\ref{lem:ConditionFunctionEquivalence} we define $f^{\Sigma}$
in the following way:
\begin{defn}
\label{def:Definition-f-Sigma}We
define the function $f^{\Sigma}:\{0,1\}^{n}\rightarrow\{0,1\}^{n}$
according to Lemma~\ref{lem:ConditionFunctionEquivalence}. I.e. for
$i\in[n]$:
\begin{equation}
f_{i}^{\Sigma}(v)=\mux_{i}(v)\oplus v_{i},\label{eq:DefinitionSkeleton}
\end{equation}
where $\mux_{i}(v):=\exists j\in[n]\backslash\{i\}:\big(\sigma_{i,j}\not=0\wedge\big(v_{i}\oplus^{\sigma_{i,j}}v_{j}\big)\big)$.
We call the graph $\aG{f^{\Sigma}}$ the skeleton of $\mG{\Sigma}$.
\end{defn}
For two nodes $v,w\in\{0,1\}^{n}$, $v\not=w$ there is a directed
path in $\aG{f^{\Sigma}}$ if and only if there is a path in $\mG{\Sigma}$.
The reduction of $\mG{\Sigma}$ to $\aG{f^{\Sigma}}$ has not only
the advantage that $\aE{f^{\Sigma}}$ can be significantly smaller
than $\mG{\Sigma}$ but also that certain structural features of $\aG{f^{\Sigma}}$
can be deduced directly from $f^{\Sigma}$. This includes attractors
\cite{Klarner_approximating_attractors,garg2008_attractors}, trap
spaces \cite{Klarner2015} and no-return sets, i.e., sets of states that no trajectory enters.

We implemented the generation of the Boolean function $f^{\Sigma}$
from $\Sigma$ into Python. The source code is available
in \url{https://github.com/RSchwieger/ASTG_from_IG}. We illustrate
the construction of $f^{\Sigma}$ with the following example:
\begin{example}
\label{exa:Sigma_f}Let's construct the functions $\mu^{\Sigma}$
and $f^{\Sigma}$ from the running example: 
\begin{align*}
\mux_{1}(v) & =\vee_{j\in\{4\}}\big(v_{1}\oplus^{\sigma_{1,j}}v_{j}\big)\\
 & =\neg(v_{1}\oplus v_{4})\\
\mux_{2}(v) & =\vee_{j\in\{1\}}\big(v_{2}\oplus^{\sigma_{2,j}}v_{j}\big)\\
 & =v_{1}\oplus v_{2}\\
\mux_{3}(v) & =\vee_{j\in\{2,4\}}\big(v_{3}\oplus^{\sigma_{3,j}}v_{j}\big)\\
 & =(v_{2}\oplus v_{3})\vee\neg(v_{3}\oplus v_{4})\\
\mux_{4}(v) & =\vee_{j\in\{3\}}\big(v_{4}\oplus^{\sigma_{4,j}}v_{j}\big)\\
 & =\neg(v_{3}\oplus v_{4})
\end{align*}
And 
\begin{align*}
f_{1}^{\Sigma}(v) & =v_{1}\oplus\mux_{1}(v)=\neg v_{4}\\
f_{2}^{\Sigma}(v) & =v_{2}\oplus\mux_{2}(v)=v_{1}\\
f_{3}^{\Sigma}(v) & =v_{3}\oplus\mux_{3}(v)=v_{3}\oplus\big((v_{2}\oplus v_{3})\vee\neg(v_{3}\oplus v_{4})\big)\\
f_{4}^{\Sigma}(v) & =v_{4}\oplus\mux_{4}(v)=\neg v_{3}
\end{align*}
In Fig.~\ref{fig:Skeleton-of-the-running-example} the skeleton of
the running example is depicted. The graph $\mG{\Sigma}$ has $50$
edges and its skeleton $\aG{f^{\Sigma}}$ has $36$
edges (without counting the two self -loops). 
\begin{figure}
\includegraphics[width=0.9\textwidth]{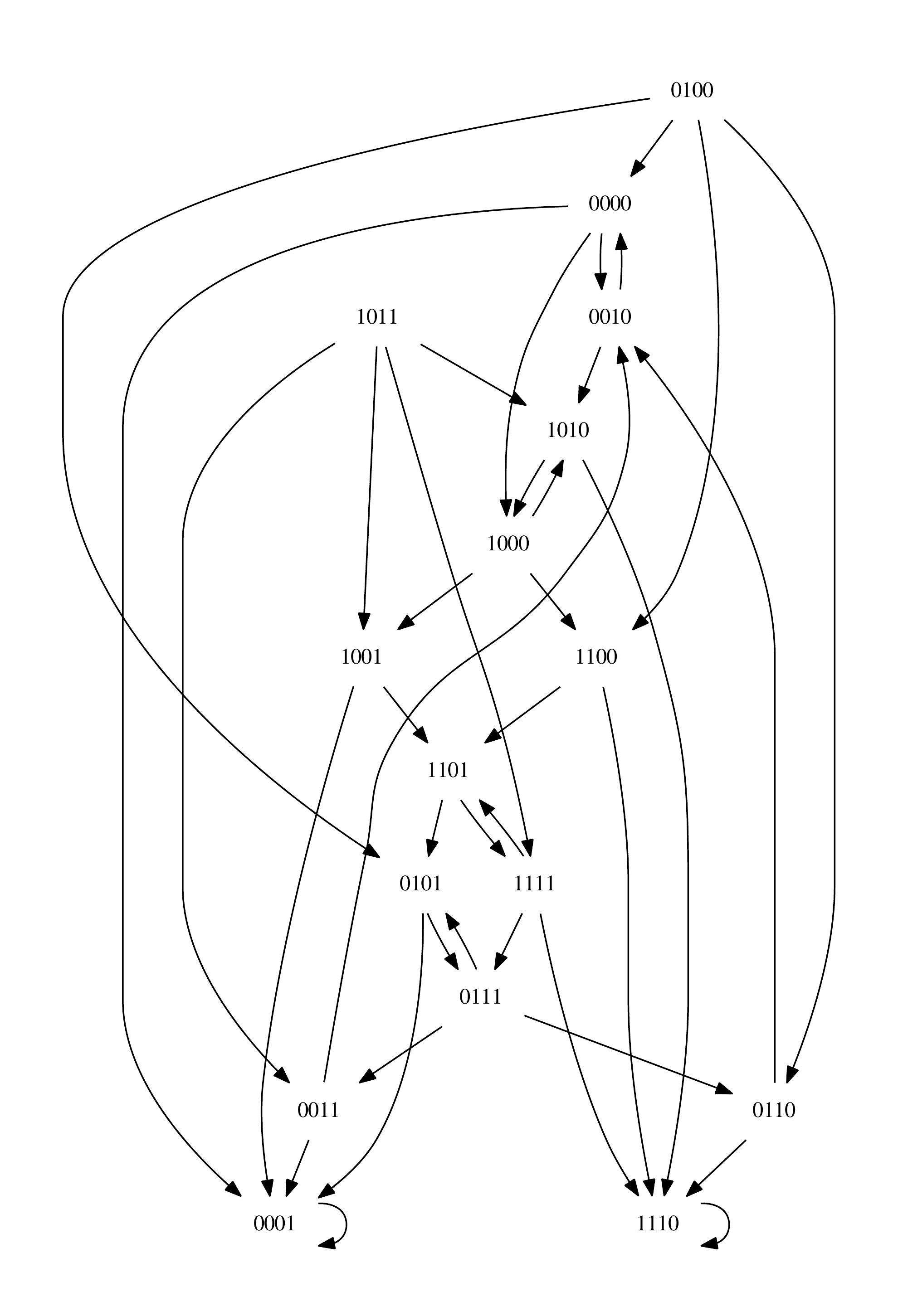}\caption{\label{fig:Skeleton-of-the-running-example}Skeleton of the running
example.}

\end{figure}
\end{example}

\section{\label{sec:Applications}Applications}

The reduction of graph $\mG{\Sigma}$ to its skeleton $\aG{f^{\Sigma}}$
allows us to use the asynchronous state transition graph to analyze
the graph $\mG{\Sigma}$. In order to demonstrate the usefulness of
the skeleton we select two possible applications.

\subsection{Consistency with time series data}

A common way to validate a model of a gene regulatory network is to test whether a given time series data is
compatible with it. Such time series data
can contain for example concentration levels of proteins or mRNA. In practice, often not all components of the regulatory network are measured and also the amount of measurements can be
relatively low compared to the size of the system. Due to the lack of precise empirical data, qualitative models are used frequently. Here it often plays only a role
which components of the system influence each other and what is the sign of their influences (activating or inhibiting).
In the framework of QDEs testing for compatibility translates to the task of testing whether the graph $\mG \Sigma$ possesses a trajectory containing
specific nodes representing our measurements. 
Thereby, the nodes of the graph $\mG \Sigma$ represent signs of activity trends of its components. The matrix $\Sigma$ contains the information
about the influences in the gene regulatory network. If the modeler does not find such a
trajectory in $\mG \Sigma$, the measurements are not compatible with the sign matrix $\Sigma$.

We can formulate the question of compatibility as a model checking query.
Model checking is a formal method from computer science that solves
the problem of deciding whether a temporal logic specification is
satisfied by a given transition system.  In the field of qualitative
differential equations it has been used to analyze the QSTG \cite{Kuipers2001}.
Our result shows that we can solve the problem of compatibility also
on the transition system induced by $f^{\Sigma}$.
However, we do not know if this offers any advantages with respect to the running time. In any case it is interesting to elucidate
in how far analysis methods and tools from these fields are related.

For a systematic introduction into model checking we
refer to \cite{Baier2008} and \cite[Chapter 3]{klarner2015contributions}.
Notions and definitions in this subsection follow roughly \cite{klarner2015contributions}.
\begin{defn}[{\cite[Def. 16]{klarner2015contributions}}]
A discrete time series is a sequence $P=(p_{0},\dots,p_{m})$ of
$m\in\mathbb{N}$ vectors $p_{i}\in\{0,1,?\}^{n}$,
where $p_{i}$ represents the i-th measurement of the experiment.
Components not measured are marked with the symbol $?$.
\end{defn}
For a vector $p\in\{0,1,?\}^{n}$ we denote with $S[p]$ the set $S[p]=\big\{ s\in\{0,1\}|\forall i\in[n]:\big((p_{i}\not=?)\rightarrow(p_{i}=s_{i})\big)\big\}$.
\begin{defn}
For a given sign matrix $\Sigma$, the graph $\mG{\Sigma}$ or $\aG{f^{\Sigma}}$
is called compatible with a time series $P=(p_{0},\dots,p_{m})$ iff
there is a sequence of states $(x_{0},\dots,x_{m})$ with $x_{i}\in S[p_{i}]$
such that for every $0\leq i<m$ there is a directed path starting
in $x_{i}$ and ending in $x_{i+1}$ in $\mG{\Sigma}$ or $\aG{f^{\Sigma}}$.
\end{defn}
The compatibility statement for a given times series $P$ can be translated
into a computation tree logic formula (CTL formula). For example the CTL
formula that queries whether the time series $(p_{0},p_{1},p_{2})$
is compatible with a given a model is defined by nesting $EF$ operators in the following way: 
\[
\phi:=p_0 \wedge EF(p_{1}\wedge EF(p_{2}))
\]
with initial state(s) defined by $p_{0}$.
The letter $E$ stands for the existential
``there is a path'' and the letter $F$ stands for the
``finally'' operators \cite[p. 31]{klarner2015contributions}.

\begin{defn}
The nested reachability query $R(P)$ for a time series $P=(p_{0},\dots,p_{m})$
is defined recursively by 
\begin{align*}
\phi_{m} & :=p_{m}\\
\phi_{m-t} & :=p_{m-t}\wedge EF\phi_{m-t+1},t=1,\dots,m
\end{align*}
and $R(P):=\phi_{0}$.
\end{defn}
Such queries can be implemented in PyBoolNet. In order to test compatibility of a time series $P$ with a given model $\Sigma$ we need
to test then if there exists a path satisfying $R(P)$.
We demonstrate this with our running example:
\begin{example}
Let us assume we are given a biological system with four species $v0,v1,v2$
and $v3$. We want to know if the matrix $\Sigma$ is a realistic
interaction graph for our biological system. Experiments suggest that there is
a trajectory from the subspace $11??$ to $0100$. We use here PyBoolNet
\cite{klarner2016pyboolnet}, which uses the model checker NuSMV. We can check this property with the CTL query
$\phi:=EF((v0 \& v1) \& EF(!v0 \& v1 \& !v2 \& !v3))$.\footnote{NuSMV uses the symbol $\&$ to denote a logical ``AND'' and $!$ to denote a negation.} Checking this query reveals that this is
not the case and therefore $\Sigma$ is not an accurate model for the considered biological system.
\end{example}

\subsection{Network inference}

Instead of testing whether a given model is compatible with time series data,
we could use the data to construct a model from scratch (network inference).
In our context this means we want to find a function $f^\Sigma$, which agrees with the given data. Then we can read off from
$f^\Sigma$ the matrix $\Sigma$. Due to the definition of $\mu^{\Sigma}$ in Eq. (\ref{eq:conditionFunction})
and $f^{\Sigma}$ in Eq. (\ref{eq:DefinitionSkeleton}), we can restrict the inference problem to the following set of Boolean functions:
\begin{align*}
I_{i}:= & \big\{ f:\{0,1\}^{n}\rightarrow\{0,1\}, x\mapsto x_{i}\oplus[\vee_{j\in A}(x_{i}\oplus^{\sigma_{j}}x_{j})]\\
 & |A\subseteq[n]\backslash\{i\},\forall j \in A:\sigma_{j}\in\{-1,1\}\big\},i\in[n],\\
I:= & \prod_{i=1}^{n}I_{i},
\end{align*}

Different approaches for solving
inference problems have been developed \cite{Kalman_filter_boolean_inference,Barman2017,inferring_boolean_networks_berestovsky2013evaluation,reveal},
which can be adopted to the set of Boolean functions $I$.

\section{Conclusion and outlook}

The graph $\mG{\Sigma}$ is useful in applications to find restrictions on the behavior of solutions
in the model ensemble induced by $\Sigma$. When analyzing the graph $\mG{\Sigma}$ one is
typically interested in statements about the reachability between nodes in $\mG{\Sigma}$.
We proved that for a reachability-analysis it is enough to consider a skeleton
$\aG{f^{\Sigma}}$, which is the ASTG of a Boolean function $f^{\Sigma}$. 
For this purpose we redefined the graph $\mG{\Sigma}$ stemming from \cite{Eisenack2006} over the state space $\{0,1\}^{n}$.
The graph $\aG{f^{\Sigma}}$
can have significantly less edges than $\mG{\Sigma}$.

This reduction allows us furthermore to use existing methods developed for Boolean networks for analyzing the graph $\mG{\Sigma}$.
We selected two applications -- consistency with time series data and network inference --
to demonstrate how this link can be used. Such potential applications should be further researched. Especially it is interesting to investigate 
in how far Boolean inference algorithms restricted
to the set $I$ perform in comparison to other Boolean inference algorithms.

\bibliographystyle{amsplain}
\bibliography{bibtex.bib}

\end{document}